\documentclass[conference,letterpaper,uplatex, dvipdfmx]{IEEEtran}
\IEEEoverridecommandlockouts

\usepackage{cite}
\usepackage{url}
\usepackage{ifthen}
\usepackage{algorithm, algorithmicx, algpseudocode}
\usepackage{graphicx}
\usepackage{textcomp}
\usepackage[cmex10]{amsmath}
\usepackage{amssymb,amsfonts}
\usepackage{pifont}
\usepackage{amsthm}
\usepackage{mathrsfs}
\def\BibTeX{{\rm B\kern-.05em{\sc i\kern-.025em b}\kern-.08em T\kern-.1667em\lower.7ex\hbox{E}\kern-.125emX}}
\usepackage{tikz}
\usetikzlibrary{automata,positioning,calc}
\usetikzlibrary{intersections}
\usetikzlibrary{decorations.pathreplacing,angles,quotes}

\usepackage{bm}
\usepackage{amscd}

\theoremstyle{definition}
\newtheorem{theorem}{Theorem}
\newtheorem{definition}{Definition}
\newtheorem{lemma}{Lemma}
\newtheorem{prop}{Proposition}
\newtheorem{cor}{Corollary}
\newtheorem{remark}{Remark}
\newtheorem{eg}{Example}

\newcommand{\abs}[1]{\left\lvert#1\right\rvert}

\newcommand{\one}[1]{\mbox{1}\hspace{-0.25em}\mbox{l}_{\left\{#1\right\}}}

\newcommand{\relmiddle}[1]{\mathrel{}\middle#1\mathrel{}}

\def\vE{\mathbb E}

\font\b=cmr10 scaled\magstep4

\def\bigzerou{\smash{\lower1.7ex\hbox{\b 0}}}
\def\bigzerou{\smash{\lower1.7ex\hbox{\b 0}}}
 
\begin{document}

\title{On Hypothesis Testing via a Tunable Loss \\
}

\author{
\IEEEauthorblockN{Akira Kamatsuka}
\IEEEauthorblockA{Shonan Institute of Technology \\ 
Email: \text{kamatsuka@info.shonan-it.ac.jp}
 }
}

\maketitle

\begin{abstract}
We consider a problem of simple hypothesis testing using 
a randomized test via a tunable loss function proposed by Liao \textit{et al}. 
In this problem, we derive results 
that correspond to the Neyman--Pearson lemma, the Chernoff--Stein lemma, 
and the Chernoff-information in the classical hypothesis testing problem. 
Specifically, we prove that the optimal error exponent of our problem in the Neyman--Pearson's setting is consistent with the classical result. 
Moreover, we provide lower bounds of the optimal Bayesian error exponent.
\end{abstract}


\section{Introduction}

Hypothesis testing is a form of statistical inference in which a 
judgment is made about a parameter $\theta\in \Theta$ or probability distribution $p_{X\mid \theta}$ 
using a random sample $X^{n}=(X_{1},\dots, X_{n})$ from the distribution, where $\Theta$ is a parameter space.
Hypotheses to be tested are expressed in the form of $H_{0}\colon \theta\in \Theta_{0}$ (null hypothesis) v.s. $H_{1}\colon \theta\in \Theta_{1}$ (alternative hypothesis) 
such that $\Theta = \Theta_{0}\cup \Theta_{1}$ and $\Theta_{0}\cap \Theta_{1} = \varnothing$. 

Commonly used criteria for a specific test are type I error and type II error introduced by Neyman and Pearson \cite{neyman_pearson_1933}. 
In \cite{1933}, they also introduced the concept of what is now called the most powerful test (MP test) and showed that the likelihood ratio test could be the MP test, 
especially for a simple hypothesis testing problem, i.e., $\Theta_{0} = \{\theta_{0}\}$ and $\Theta_{1} = \{\theta_{1}\}$. 
In the simple hypothesis problem, Stein analyzed the asymptotic error exponent in his unpublished paper, 
the results of which were later organized by Chernoff in \cite{10.1214/aoms/1177728347}. 
While in the Bayesian setting, the asymptotic Bayesian error exponent is characterized by the Chernoff-information introduced in \cite{10.1214/aoms/1177729330} (see \cite{Cover:2006:EIT:1146355}). 

It is well known that statistical inferences, including hypothesis testing and estimation, can be formulated via the   
statistical decision-theoretic framework developed by Wald \cite{10.1214/aoms/1177730030}, 
using the concept of decision function and loss function. In particular, classical hypothesis testing problems can be formulated using deterministic test functions and $0$-$1$ loss.

Recently, Liao \textit{et al.} have introduced a tunable loss function called $\alpha$-loss\footnote{Note that we later use $\nu$ instead of $\alpha$ as the notation for the tunable parameter to avoid confusion with the type I error notation.}, 
$\alpha\in [1, \infty]$ for estimation problems in \cite{8804205}, which can represent 
log-loss $(\alpha=1)$ and the soft $0$-$1$ loss $(\alpha = \infty)$, to model an adversary in the  privacy-preserving data publishing problem. 
This tunable loss function can be interpreted as a loss function for a randomized decision function for the estimation problem and has been applied to a variety of problems in recent years, such as binary classification \cite{8849796,9174356,9761766}, generative adversarial network (GAN) \cite{9611499},\cite{https://doi.org/10.48550/arxiv.2205.06393} and guessing \cite{9517733}.

In this work, we consider a simple hypothesis testing problem using randomized test functions and apply the tunable loss function for this problem.

Our main contributions are as follows:
\begin{itemize}
\item In the Neyman--Pearson setting, we derive the optimal randomized test 
corresponding to the Neyman--Pearson lemma (Theorem \ref{thm:nu_MP_test}). 
We also characterize the optimal error exponent (Theorem \ref{thm:nu_opt_exponent}), 
which is consistent with the classical result called Chernoff--Stein lemma.
\item In the Bayesian setting, we derive the optimal randomized test (Proposition \ref{prop:opt_bayes_ht}) and 
provide lower bounds of the optimal Bayesian error exponent that depend on the tunable parameter (Theorem \ref{thm:nu_bayes_opt_exponent}). 
Note that these lower bounds correspond to the Chernoff-information in the classical hypothesis testing.
\end{itemize} 

\section{Preliminary}\label{sec:preliminary}
We first review the simple hypothesis testing problem and the tunable loss function 
via the statistical decision theory \cite{GVK027440176}. 
We will assume that all alphabets are finite.

\subsection{Simple hypothesis testing}
Simple hypothesis testing (or binary hypothesis testing) is a statistical inference 
to make a conclusion about hypotheses on a parameter $\theta\in \Theta=\{\theta_{0}, \theta_{1}\}$ of a 
population probability distribution $p_{X\mid\theta}$ of the form $H_{0}\colon \theta = \theta_{0}$ v.s. $H_{1}\colon \theta = \theta_{1}$ 
using a random sample $X^{n} = (X_{1}, \dots, X_{n})\in \mathcal{X}^{n}, X_{i}\in \mathcal{X}, i=1,\dots, n$.
Let $A = \delta_{n}(X^{n})\in  \{0, 1\}$ be the decision made from a random sample $X^{n}$, 
where $\delta_{n}\colon \mathcal{X}^{n}\to \{0, 1\}$ is a deterministic test function and 
$A=1$ means rejecting the null hypothesis $H_{0}$ while $A=0$ means accepting $H_{0}$. 
Let $\Delta_{n}$ be a set of all deterministic test functions.

In the statistical decision-theoretic formulation of the simple hypothesis testing, 
the following loss function $\ell(\theta, a)$ and risk function 
$R(\theta, \delta_{n}):= \vE_{X^{n}\mid \theta}\left[\ell(\theta, \delta_{n}(X^{n}))\right] = \sum_{x^{n}}p_{X^{n}\mid \theta}(x^{n}\mid \theta)\ell(\theta, \delta_{n}(x^{n}))$ are used.

\begin{definition}[Loss and risk function]
\begin{align}
\ell(\theta, a) &:= 
\begin{cases}
1-\one{0}(a), & \theta = \theta_{0}, \\ 
1-\one{1}(a), & \theta = \theta_{1}, 
\end{cases} \\ 
R(\theta, \delta_{n}) 
&:= \vE_{X^{n}\mid \theta}\left[\ell(\theta, \delta_{n}(X^{n}))\right] \\
&=
\begin{cases}
\alpha(\theta_{0}, \delta_{n}), & \theta = \theta_{0}, \\ 
\bar{\beta}(\theta_{1}, \delta_{n}), & \theta=\theta_{1}, 
\end{cases}, 
\end{align} 
where $\one{i}(a), i=0,1$ is an indicator function of a singleton $\{i\}$ and $\alpha(\theta_{0}, \delta_{n}), \bar{\beta}(\theta_{1}, \delta_{n})$ are type I error and type II error, respectively, defined as follows:
\begin{align}
\alpha(\theta_{0}, \delta_{n}) &:= \vE_{X^{n}\mid \theta_{0}}\left[\delta_{n}(X^{n})\right] \\
&= \sum_{x^{n}} p_{X^{n}\mid \theta_{0}}(x^{n}\mid \theta_{0}) \delta_{n}(x^{n}),  \quad \text{(Type I error)}  \\
\bar{\beta}(\theta_{1}, \delta_{n}) &:= 1-\vE_{X^{n}\mid \theta_{1}}\left[\delta_{n}(X^{n})\right] \\ 
&= 1- \sum_{x^{n}} p_{X^{n}\mid \theta_{1}}(x^{n}\mid \theta_{1}) \delta_{n}(x^{n}). 
\quad  \text{(Type II error)}
\end{align}
\end{definition}

In the Neyman--Pearson setting, the following MP test is considered, and they 
showed that the likelihood ratio test could be the MP test which is known as Neyman--Pearson lemma.

\begin{definition}[MP test of size $\epsilon$]
Let $\epsilon \in (0, 1)$. The test function $\delta_{n}^{\text{MP}}\colon \mathcal{X}^{n}\to \{0, 1\}$ is the MP test of size $\epsilon$ if 
the following hold:

\begin{enumerate}
\item $\alpha(\theta_{0}, \delta_{n}^{\text{MP}}) \leq \epsilon$. 
\item For any test function $\delta_{n}\in \Delta_{n}$, \\ $\alpha(\theta_{0}, \delta_{n})\leq \epsilon \Longrightarrow \bar{\beta}(\theta_{1}, \delta_{n}) \geq \bar{\beta}(\theta_{1}, \delta_{n}^{\text{MP}})$.
\end{enumerate}
\end{definition}

\begin{prop}[Neyman--Pearson lemma, \text{\cite[Thm 11.7.1]{Cover:2006:EIT:1146355}}]
The likelihood ratio test $\delta_{n}^{\text{LR}}\colon \mathcal{X}^{n}\to \{0, 1\}$ defined as follows is the MP test of size $\epsilon$:
\begin{align}
\delta_{n}^{\text{LR}}(x^{n}) &:= 
\begin{cases}
1, & \frac{p_{X^{n}\mid \theta}(x^{n}\mid \theta_{0})}{p_{X^{n}\mid \theta}(x^{n}\mid \theta_{1})} \leq \lambda, \\ 
0, & \text{otherwise}, 
\end{cases} \label{eq:LR_test}
\end{align}
where threshold $\lambda$ is defined such that $\alpha(\theta_{0}, \delta^{\text{LR}}_{n})= \epsilon$. 
\end{prop}

Stein and Chernoff characterized the optimal error exponent as follows.

\begin{definition}[$\epsilon$-optimal error exponent]
Let $\epsilon\in (0, 1)$. The $\epsilon$-optimal error exponent $B_{\epsilon}$ is defined as 
\begin{align}
B_{\epsilon} &:= -\lim_{n\to \infty} \frac{1}{n}\log \min_{\delta_{n}\colon \alpha(\theta_{0}, \delta_{n}) < \epsilon} \bar{\beta}(\theta_{1}, \delta_{n}),
\end{align}
where minimum is over all deterministic test functions $\delta_{n}$ satisfying $\alpha(\theta_{0}, \delta_{n}) < \epsilon$. 
\end{definition}

\begin{prop}[Chernoff--Stein lemma, \text{\cite[Thm 11.8.3]{Cover:2006:EIT:1146355}}] 
For any $\epsilon\in (0, 1)$, 
\begin{align}
B_{\epsilon} &= D(p_{X\mid \theta_{0}} || p_{X\mid \theta_{1}}), 
\end{align}
where $D(p_{X\mid \theta_{0}} || p_{X\mid \theta_{1}}) := \sum_{x} p_{X\mid \theta}(x\mid \theta_{0}) \log p_{X\mid \theta}(x\mid \theta_{0})/p_{X\mid \theta}(x\mid \theta_{1})$ 
is the Kullback--Leibler divergence. 
\end{prop}

In the Bayesian setting, let $\pi$ denote a prior probability on $\theta$ such that 
$\pi(\theta_{0}) = \pi_{0}, \pi(\theta_{1}) = \pi_{1}$, and $\pi_{0} + \pi_{1} =1$. 
Then the minimal Bayesian error probability and 
the optimal Bayesian error exponent are defined and characterized as follows.

\begin{definition}[Bayesian error probability]
The Bayes risk for the simple hypothesis testing or the \textit{Bayesian error probability} is defined as
\begin{align}
r(\delta_{n}) &:= \vE_{\theta}[R(\theta, \delta^{n})] \\ 
&= \pi_{0} \alpha(\theta_{0}, \delta_{n}) + \pi_{1} \bar{\beta}(\theta_{1}, \delta_{n}).
\end{align}
\end{definition}

\begin{prop}
The minimal Bayesian error probability is given by 
\begin{align}
\min_{\delta_{n}} r(\delta_{n}) &= r(\delta_{n}^{\text{Bayes}}),
\end{align}
where the optimal test function $\delta_{n}^{\text{Bayes}}\colon \mathcal{X}^{n}\to \{0, 1\}$ is given by 
\begin{align}
\delta_{n}^{\text{Bayes}}(x^{n}) &:= 
\begin{cases}
1, &  \frac{p_{X^{n}\mid \theta}(x^{n}\mid \theta_{0})}{p_{X^{n}\mid \theta}(x^{n}\mid \theta_{1})} \leq \frac{\pi_{1}}{\pi_{0}}, \\ 
0, & \text{otherwise}.
\end{cases}
\end{align}
\end{prop}

\begin{definition}[The optimal Bayesian error exponent]
The optimal Bayesian error exponent $D^{*}$ is defined as 
\begin{align}
D^{*} &:= -\lim_{n\to \infty} \frac{1}{n} \log \min_{\delta_{n}\in \Delta_{n}} r(\delta_{n}). \label{eq:Bayesian_error_exponent}
\end{align}
\end{definition}

\begin{prop}[\text{\cite[Thm 11.9.1]{Cover:2006:EIT:1146355}}]
\begin{align}
D^{*} &= C(p_{X\mid \theta_{0}}, p_{X\mid \theta_{1}}),
\end{align}
where 
\begin{align}
&C(p_{X\mid \theta_{0}}, p_{X\mid \theta_{1}}) \notag \\
&:= -\min_{0\leq \lambda \leq 1} \log \sum_{x} p_{X\mid \theta}(x\mid \theta_{0})^{\lambda}p_{X\mid \theta}(x\mid \theta_{1})^{1-\lambda} \label{eq:chernoff_information}
\end{align} 
is the Chernoff-information.
\end{prop}

\subsection{Tunable loss for point estimation}
In this subsection, $\Theta$ is a finite parameter space whose number of elements is greater than or equal to $2$. 
Liao \textit{et al.} introduced the following tunable loss function, which can represent 
log-loss and the soft $0$-$1$ loss as special cases \cite{8804205}. 
From the statistical decision-theoretic perspective, this tunable loss function can be interpreted as 
a loss function for a randomized decision rule $\delta_{n, \text{est}}^{*}\colon \mathcal{X}^{n}\times \Theta \to [0, 1]$ for point estimation of 
parameter $\theta\in \Theta$.

\begin{definition}[\text{$\nu$-loss \cite[Def 3]{8804205}}] \label{def:alpha_loss}
\begin{align}
&L_{\nu}(\theta, \delta^{*}_{n, \text{est}}(x^{n}, \cdot)) \notag \\ 
&:= 
\begin{cases}
-\log \delta^{*}_{n, \text{est}}(x^{n}, \theta),  & \nu=1, \\
\frac{\nu}{\nu-1} \left\{ 1 - \delta^{*}_{n, \text{est}}(x^{n}, {\theta})^{\frac{\nu-1}{\nu}} \right\}, & \nu > 1, \\ 
1 - \delta^{*}_{n, \text{est}}(x^{n}, \theta), &  \nu = \infty.
\end{cases}
\end{align}
\end{definition}
Figure \ref{fig:graph_nu_loss} shows the tunable function for different values of $\nu$. 
\begin{remark}
Originally, Liao \textit{et al.} used $\alpha$ as a notation for the tunable parameter. However, we use $\nu$ instead to avoid confusion with the type I error notation.
\end{remark}

\begin{figure}[htbp]
\centering
\includegraphics[width=3.5in, clip]{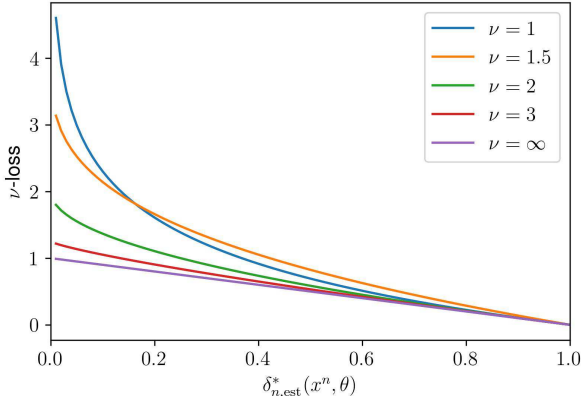}
\caption{The plot of the $\nu$-loss $L_{\nu}(\theta, \delta^{*}_{n, \text{est}}(x^{n}, \cdot))$}
\label{fig:graph_nu_loss}
\end{figure}

\begin{prop}[\text{\cite[Lem 1]{8804205}}] \label{prop:Bayes_estimation}
Let $\pi$ be a prior distribution on $\theta\in \Theta$ and $r(\delta^{*}_{n, \text{est}}) := \vE_{X^{n}, \theta}\left[L_{\nu}(\theta, \delta^{*}_{n, \text{est}}(X^{n}, \cdot))\right]=\sum_{x^{n}, \theta}\pi(\theta)p_{X^{n}\mid \theta}(x^{n}\mid \theta)L_{\nu}(\theta, \delta_{n, \text{est}}^{*}(x^{n}, \cdot))$ be Bayes risk for 
point estimation using a randomized decision function $\delta^{*}_{n, \text{est}}$. 
Then, the minimal Bayes risk is given by 
\begin{align}
\inf_{\delta^{*}_{n, \text{est}}} r(\delta^{*}_{n, \text{est}}) &= r(\delta_{n, \text{est}}^{*, \text{Bayes}}),
\end{align}
where the optimal randomized decision function $\delta_{n, \text{est}}^{*, \text{Bayes}}\colon \mathcal{X}^{n}\times \Theta\to [0, 1]$ is 
given as follows: 

For $\nu\in [1, \infty)$, 
\begin{align}
\delta^{*, \text{Bayes}}_{n, \text{est}}(x^{n}, \hat{\theta}) &= 
\begin{cases}
\pi(\hat{\theta}\mid x^{n}), & \nu=1, \\ 
\frac{\pi(\hat{\theta}\mid x^{n})^{\nu}}{\sum_{\hat{\theta}} \pi(\hat{\theta} \mid x^{n})^{\nu}} & \nu \in (1, \infty), 
\end{cases}
\end{align}
where $\pi(\hat{\theta}|x^{n}) := \pi(\hat{\theta})p_{X^{n}\mid \theta}(x^{n}\mid \hat{\theta})/\sum_{\theta} \pi(\theta)p_{X^{n}\mid \theta}(x^{n}\mid \theta)$ is a posterior distribution on $\theta$ given $X^{n} = x^{n}$. 
For $\nu = \infty$, 
\begin{align}
\delta^{*, \text{Bayes}}_{n}(x^{n}, \hat{\theta}) &= 
\begin{cases}
1/\abs{\text{MAP}(x^{n})}, & \hat{\theta} \in \text{MAP}(x^{n}), \\ 
0, & \text{otherwise},
\end{cases} 
\end{align}
where $\text{MAP}(x^{n}) := \left\{\, \tilde{\theta}\in \Theta \relmiddle{|} \pi(\tilde{\theta}\mid x^{n}) = \max_{\theta}\pi(\theta\mid x^{n}) \right\}$. 
\end{prop}

\section{Hypothesis Testing via a Tunable-loss}\label{sec:ht_under_P_loss}
In this section, we will formulate hypothesis testing problems via the tunable-loss in both Neyman--Peason's setting and the Bayesian setting.

Let $\Theta=\{\theta_{0}, \theta_{1}\}$ be a parameter space, 
$\delta_{n}^{*}\colon \mathcal{X}^{n}\times \{0, 1\}\to [0, 1]$ be a randomized test function for hypotheses
$H_{0}\colon \theta = \theta_{0}$ v.s. $H_{1}\colon \theta=\theta_{1}$, i.e., given $X^{n} =x^{n}$, 
$\delta_{n}^{*}(x^{n}, 0)$ represents \textit{the probability of accepting the null hypothesis $H_{0}$} and 
$\delta_{n}^{*}(x^{n}, 1) = 1-\delta_{n}^{*}(x^{n}, 0)$ represents \textit{the probability of accepting the alternative hypothesis $H_{1}$ ($=$ probability of rejecting the null hypothesis $H_{0}$)}, respectively. Let $\Delta_{n}^{*}$ be a set of all randomized test functions. 
In this section, we formulate problems of the simple hypothesis testing 
using the randomized test and a tunable loss function. 
First, we define $\nu$-loss for the hypothesis testing as follows.

\begin{definition}[$\nu$-loss for a test function]
For $\nu \in (1, \infty)$, $\nu$-loss for a randomized test $\delta_{n}^{*}$ is defined as follows:
\begin{align}
L_{{\nu}}(\theta, \delta_{n}^{*}(x^{n}, \cdot)) &= 
\begin{cases}
\frac{\nu}{\nu-1} \left\{1 - \delta^{*}_{n}(x^{n}, 0)^{\frac{\nu-1}{\nu}} \right\}, & \theta = \theta_{0}, \\ 
\frac{\nu}{\nu-1} \left\{1 - \delta^{*}_{n}(x^{n}, 1)^{\frac{\nu-1}{\nu}} \right\}, & \theta = \theta_{1}.
\end{cases}
\end{align}
The value of $L_{\nu}(\theta, \delta_{n}^{*}(x^{n}, \cdot))$ is extended by continuity to 
$\nu=1$ (log-loss) and $\nu=\infty$ (soft $0$-$1$ loss) as Definition \ref{def:alpha_loss}. 
\end{definition}

\begin{remark}
The higher the probability $\delta_{n}^{*}(x^{n}, 0)$ of accepting the null hypothesis $H_{0}\colon \theta=\theta_{0}$ when it is correct, the smaller the value of the loss function. 
Similarly, the higher the probability $\delta_{n}^{*}(x^{n}, 1) = 1-\delta_{n}^{*}(x^{n}, 0)$ of rejecting the null hypothesis when it is false, the smaller the value of the loss function.
\end{remark}

Then we define $\nu$-Type I/II error and $\nu$-Bayesian error probability via 
risk function $R(\theta, \delta_{n}^{*}):= \vE_{X^{n}\mid \theta}\left[L_{\nu}(\theta, \delta_{n}^{*}(X^{n}, \cdot))\right]$ and Bayes risk function $r(\delta_{n}^{*}):=\vE_{\theta}\left[R(\theta, \delta_{n}^{*})\right]$.

\begin{definition}[$\nu$-Type I/II error, $\nu$-Bayesian error]
\begin{align}
R(\theta, \delta^{*}_{n}) 
&=
\begin{cases}
\alpha_{{\nu}}(\theta_{0}, \delta^{*}_{n}), & \theta = \theta_{0}, \\ 
\bar{\beta}_{{\nu}}(\theta_{1}, \delta^{*}_{n}), & \theta=\theta_{1}, 
\end{cases} \\
r(\delta_{n}^{*}) &= \pi_{0}\alpha_{{\nu}}(\theta_{0}, \delta^{*}_{n})  + \pi_{1}\bar{\beta}_{{\nu}}(\theta_{1}, \delta^{*}_{n}), (\text{{$\nu$}-Bayesian error}) 
\end{align} 
where $\pi_{i} = \pi(\theta_{i}), i=0,1$ is the prior probability on $\theta=\theta_{i}$ and 
\begin{align}
\alpha_{{\nu}}(\theta_{0}, \delta^{*}_{n}) 
&:=  \frac{\nu}{\nu-1} \left\{ 1 - \vE_{X^{n}\mid \theta_{0}} \left[\delta^{*}_{n}(X^{n}, 0)^{\frac{\nu-1}{\nu}}\right] \right\}, \notag \\ 
&\qquad \qquad \qquad \qquad \qquad(\text{{$\nu$}-Type I error}) \\
\bar{\beta}_{{\nu}}(\theta_{1}, \delta^{*}_{n}) 
&:=  \frac{\nu}{\nu-1} \left\{ 1 - \vE_{X^{n}\mid \theta_{1}} \left[\delta^{*}_{n}(X^{n}, 1)^{\frac{\nu-1}{\nu}}\right] \right\}. \notag \\ 
&\qquad \qquad \qquad \qquad \qquad(\text{{$\nu$}-Type II error})
\end{align}
\end{definition}

Based on the $\nu$-Type I/II error and $\nu$-Bayesian error,  
we extend the concepts in the classical hypothesis testing problem as follows.

\begin{definition}[$\nu$-MP test of size $\epsilon$] 
Let $\epsilon\in (0, 1)$. The randomized test function $\delta_{n}^{*, \text{$\nu$-\text{MP}}}\colon \mathcal{X}^{n}\times \{0, 1\}\to [0, 1]$ is the $\nu$-MP test of size $\epsilon$ if the following hold:
\begin{enumerate}
\item $\alpha_{{\nu}}(\theta_{0}, \delta_{n}^{*, {\nu}\text{-MP}}) \leq \epsilon$, 
\item For any randomized test $\delta^{*}_{n}\in \Delta_{n}^{*}$, \\
$\alpha_{{\nu}}(\theta_{0}, \delta^{*}_{n})\leq \epsilon \Longrightarrow 
\bar{\beta}_{{\nu}}(\theta_{1}, \delta^{*}_{n}) 
\geq \bar{\beta}_{{\nu}}(\theta_{1}, \delta_{n}^{*, {\nu}\text{-MP}})$.
\end{enumerate}
\end{definition}

\begin{definition}[$(\nu, \epsilon)$-error exponent]
Let $\nu\in [1, \infty]$ and $\epsilon\in (0, 1)$.
The $(\nu, \epsilon)$-optimal error exponent $B_{\nu, \epsilon}$ is defined as
\begin{align}
B_{\nu, \epsilon} &:= - \lim_{n\to \infty}\frac{1}{n} \log \inf_{\delta_{n}^{*}\colon \alpha_{\nu}(\theta_{0}, \delta_{n}^{*}) < \epsilon} \bar{\beta}(\theta_{1}, \delta_{n}^{*}),
\end{align}
where infimum is over all randomized test functions $\delta_{n}^{*}$ satisfying $\alpha_{\nu}(\theta_{0}, \delta_{n}^{*}) < \epsilon$. 
\end{definition}

\begin{definition}[$\nu$-Bayesian error exponent]
Let $\nu\in [1, \infty]$. 
The $\nu$-Bayesian error exponent $D_{\nu}^{*}$ is defined as 
\begin{align}
D_{\nu}^{*} &:= -\lim_{n\to \infty} \frac{1}{n} \log \inf_{\delta_{n}^{*}\in \Delta_{n}^{*}} r(\delta_{n}^{*}).
\end{align}
\end{definition}

\section{Main Results}\label{sec:main_result}
The main results of this paper are 
derivations of the $\nu$-MP test, characterization of the $(\nu, \epsilon)$-error exponent, 
and derivation of lower bounds of the $\nu$-Bayesian error exponent.

\begin{theorem} \label{thm:nu_MP_test}
Let $\epsilon\in (0, 1)$. The $\nu$-MP test of size $\epsilon$ is given as follows: \\
For $\nu\in [1, \infty)$, 
\begin{align}
\delta^{*, \text{$\nu$-MP}}(x^{n}, 1) 
&= 
\frac{\lambda^{-{\nu}}p_{X^{n}\mid \theta}(x^{n} \mid \theta_{0})^{-{\nu}}}{p_{X^{n}\mid \theta}(x^{n}\mid \theta_{1})^{-{\nu}} + \lambda^{-{\nu}}p_{X^{n}\mid\theta}(x^{n}\mid \theta_{0})^{-{\nu}}}, \label{eq:nu_MP_test_1}\\
\delta^{*, \text{$\nu$-MP}}(x^{n}, 0) &= 1- \delta^{*, \text{$\nu$-MP}}(x^{n}, 1). \label{eq:nu_MP_test_0}
\end{align}
For $\nu  = \infty$, 
\begin{align}
\delta^{*, \text{$\infty$-MP}}(x^{n}, 1) &= 
\begin{cases}
1, & \frac{p_{X^{n}\mid \theta}(x^{n}\mid \theta_{0})}{p_{X^{n}\mid \theta}(x^{n}\mid \theta_{1})}\leq \lambda \\ 
0, & \text{otherwise},
\end{cases} \label{eq:infty_MP_test_1} \\ 
\delta^{*, \text{$\infty$-MP}}(x^{n}, 0) &= 1- \delta^{*, \text{$\infty$-MP}}(x^{n}, 1).  \label{eq:infty_MP_test_2}
\end{align}
Note that $\lambda$ is determined such that $\alpha_{\nu}(\theta_{0}, \delta_{n}^{\text{$\nu$-MP}})= \epsilon$ for $\nu\in [1, \infty]$.
\end{theorem}
\begin{remark}
Note that the $\infty$-MP test \eqref{eq:infty_MP_test_1}\eqref{eq:infty_MP_test_2} corresponds to the likelihood test defined in \eqref{eq:LR_test}. 
\end{remark}
\begin{proof}
See Appendix \ref{proof:nu_MP_test}. 
\end{proof}

\begin{theorem} \label{thm:nu_opt_exponent}
For any $\epsilon\in (0, 1)$ and any $\nu\in [1, \infty]$, 
\begin{align}
B_{{\nu}, \epsilon} &= D(p_{X\mid \theta_{0}} || p_{X\mid \theta_{1}}), 
\end{align}
where $D(p_{X\mid \theta_{0}} || p_{X\mid \theta_{1}})$ is the Kullback--Leibler divergence. 
\end{theorem}
\begin{remark}
The $(\nu, \epsilon)$-error exponent does not depend on $\nu$ as well as $\epsilon$.
\end{remark}
\begin{proof}
See Appendix \ref{proof:nu_opt_exponent}. 
\end{proof}

\begin{prop} \label{prop:opt_bayes_ht}
The minimal $\nu$-Bayesian error probability is given by 
\begin{align*}
\inf_{\delta_{n}^{*}} r(\delta_{n}^{*}) = r(\delta_{n}^{*, \text{Bayes}}), 
\end{align*}
where the optimal randomized test function $\delta_{n}^{*, \text{Bayes}}\colon \mathcal{X}^{n}\times \{0, 1\} \to [0, 1]$ is given as follows: \\
For $\nu\in [1, \infty)$,  
\begin{align}
\delta_{n}^{*, \text{Bayes}}(x^{n}, 0) 
&:= \frac{\pi(\theta_{0}\mid x^{n})^{\nu}}{\pi(\theta_{0}\mid x^{n})^{\nu} + \pi(\theta_{1} \mid x^{n})^{\nu}}, \\ 
\delta_{n}^{*, \text{Bayes}}(x^{n}, 1) &:= 1- \delta_{n}^{*, \text{Bayes}}(x^{n}, 0).
\end{align}
For $\nu=\infty$, 
\begin{align}
\delta_{n}^{*, \text{Bayes}}(x^{n}, 0) &= 
\begin{cases}
1, & \pi(\theta_{0}\mid x^{n}) \geq \pi(\theta_{1}\mid x^{n}), \\ 
0, & \text{otherwise}, 
\end{cases} \\
\delta_{n}^{*, \text{Bayes}}(x^{n}, 1) &= 1-\delta_{n}^{*, \text{Bayes}}(x^{n}, 0), 
\end{align}
where $\pi(\theta_{i} \mid x^{n}):=\pi_{i}p_{X^{n}\mid \theta}(x^{n}\mid \theta_{i})/\sum_{i=0}^{1}\pi_{i}p(x^{n}\mid \theta_{i})$. 
\end{prop}
\begin{proof}
It can be proved in a similar way as Proposition \ref{prop:Bayes_estimation} 
(see \cite[Appendix A]{8804205}).
\end{proof}

\begin{theorem} \label{thm:nu_bayes_opt_exponent}
\begin{align}
D_{\nu}^{*} & 
\begin{cases}
\geq D_{\text{B}, \nu}(p_{X\mid \theta_{0}}, p_{X\mid \theta_{1}}), & \nu \in [1, \infty), \\ 
= C(p_{X\mid \theta_{0}}, p_{X\mid \theta_{1}}), & \nu=\infty, 
\end{cases}
\end{align}
where $D_{\text{B}, \nu}(p_{X\mid \theta_{0}}, p_{X\mid \theta_{1}})$ is defined as 
\begin{align}
&D_{\text{B}, \nu}(p_{X\mid \theta_{0}}, p_{X\mid \theta_{1}}) := -\log \max\{\text{BC}_{{\nu}/{2}}, \text{BC}_{1-{\nu}/{2}}\}, \\
&\text{BC}_{{\nu}/{2}} := \sum_{x} {p_{X\mid \theta}(x\mid \theta_{0})^{\frac{\nu}{2}}p_{X\mid \theta}(x\mid\theta_{1})^{1-\frac{\nu}{2}}}, 
\end{align}
and $C(p_{X\mid \theta_{0}}, p_{X\mid \theta_{1}})$ is the Chernoff-information defined in \eqref{eq:chernoff_information}. 
\end{theorem}

\begin{remark}
$\text{BC}_{{\nu}/{2}}$ is called the $\frac{\nu}{2}$-skewed Bhattacharyya affinity coefficient \cite{https://doi.org/10.48550/arxiv.2207.03745}. 
Note that $D_{\text{B}, 1}$ and $\text{BC}_{1/2}$ equal the Bhattacharyya distance 
and the Bhattacharyya coefficient, respectively.
\end{remark}

\begin{proof}
See Appendix \ref{proof:nu_bayes_opt_exponent}. 
\end{proof}

Basic properties of the lower bound $D_{\text{B}, \nu}(p_{X\mid \theta_{0}}, p_{X\mid \theta_{1}})$ follows immediately from \cite[Exercise 2.28]{Han2002} and \cite[Cor 2]{6832827}. 

\begin{cor} The following hold:
\begin{enumerate}
\item $D_{\text{B}, \nu}(p_{X\mid \theta_{0}}, p_{X\mid \theta_{1}})$ is concave in $\nu\in [1, \infty)$.
\item $D_{\text{B}, \nu}(p_{X\mid \theta_{0}}, p_{X\mid \theta_{1}}) \geq 0$ for $\nu\in [1, 2)$, 
\item $D_{\text{B}, \nu}(p_{X\mid \theta_{0}}, p_{X\mid \theta_{1}}) = 0$ for $\nu=2$, 
\item $D_{\text{B}, \nu}(p_{X\mid \theta_{0}}, p_{X\mid \theta_{1}}) < 0$ for $\nu\in (2, \infty)$.
\end{enumerate}
\end{cor}
Note that when $\nu\geq 2$, this lower bound $D_{\text{B}, \nu}(p_{X\mid \theta_{0}}, p_{X\mid \theta_{1}})$ is useless.

\begin{eg}
Let $X^{n}=(X_{1},\dots, X_{n})$  be a random sample of size $n$ from the Bernoulli distribution $\textsf{Bern}(\theta)$. 
Now, consider the following hypothesis test: 
\begin{align*}
H_{0}: \theta = 0.5 \text{ vs. } H_{1}: \theta = 0.7.
\end{align*}
In this situation, Figure \ref{fig:lb_nu} shows a graph of the lower bound $D_{\text{B}, \nu}$ for $\nu\in [1, 2]$. 

\begin{figure}[htbp]
\centering
\includegraphics[width=3.5in, clip]{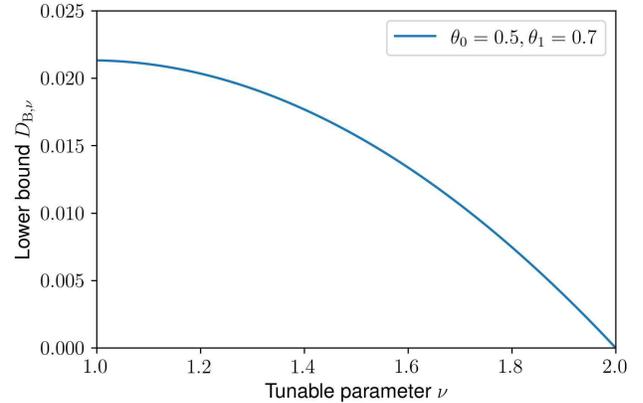}
\caption{The plot of the lower bound $D_{\text{B}, \nu}$ as a function of $\nu$}
\label{fig:lb_nu}
\end{figure}
\end{eg}

\section{Conclusion}\label{sec:conclusion}
In this work, we have developed simple hypothesis testing problems via the tunable loss by Liao \textit{et al.}\cite{8804205} 
in the Neyman--Pearsons' and Bayesian settings. 
Our results correspond to Neyman--Pearson lemma, Chernoff-Stein lemma, and Chernoff-information.
Future work includes deriving upper bound of $\nu$-Bayesian error exponent $D^{*}_{\nu}$ and 
valid lower bound for $\nu \geq 2$. 

\appendices

\section{Proof of Theorem \ref{thm:nu_MP_test}} \label{proof:nu_MP_test}
\begin{proof}
For $\nu=\infty$, it can be proved in the same way as for the Neyman--Pearson lemma by considering the set 
$\textsf{A}(\delta^{*, \text{$\infty$-MP}}_{n}):=\left\{\, x^{n}\in \mathcal{X}^{n}\relmiddle{|} \delta^{*, \text{$\infty$-MP}}_{n}(x^{n}, 0) = 1\right\}$ 
to be an acceptance region. 
For $\nu\in (1, \infty)$, the $\nu$-MP test $\delta^{*, \text{$\nu$-MP}}$ is the solution to the following optimization problem:

\begin{align*}
&\text{minimize } \bar{\beta}_{{\nu}}(\theta_{1}, \delta^{*}_{n}) 
=  \frac{\nu}{\nu-1} \left\{ 1 - \vE_{X^{n}\mid \theta_{1}} \left[\delta^{*}_{n}(X^{n}, 1)^{\frac{\nu-1}{\nu}}\right] \right\} \\
&\text{subject to }  \\
&\alpha_{{\nu}}(\theta_{0}, \delta^{*}_{n}) 
=  \frac{\nu}{\nu-1} \left\{ 1 - \vE_{X^{n}\mid \theta_{0}} \left[\delta^{*}_{n}(X^{n}, 0)^{\frac{\nu-1}{\nu}}\right] \right\} \leq \epsilon \\ 
& 0\leq \delta^{*}_{n}(x^{n}, 0) \leq 1, \\ 
& 0\leq \delta^{*}_{n}(x^{n}, 1) \leq 1, \\ 
& \delta^{*}_{n}(x^{n}, 0) + \delta^{*}_{n}(x^{n}, 1) = 1.
\end{align*}
By using the KKT conditions, we obtain \eqref{eq:nu_MP_test_1}\eqref{eq:nu_MP_test_0}. 
For $\nu=1$, it can be proved in a similar way as $\nu\in (1, \infty)$. 
\end{proof}

\section{Proof of Theorem \ref{thm:nu_opt_exponent}} \label{proof:nu_opt_exponent}
\begin{proof}
We will prove only for $\nu\in (1, \infty)$. 
Proofs for $\nu=1$ and $\nu=\infty$ can be obtained in a similar way.
For simplicity, we will denote $D(p_{X\mid \theta_{0}} || p_{X\mid \theta_{1}})$ by $D(p_{0} || p_{1})$. 

\noindent
(Direct part): 
Let $\epsilon^{\prime}>0$ be an arbitrarily small number such that 
$0 < \epsilon^{\prime} < \epsilon$ and $\epsilon^{\prime}_{{\nu}} := \epsilon^{\prime}({\nu}-1)/{\nu}$.
Define a set of $\epsilon_{\nu}^{\prime}$-\textit{relative typical sequences} 
$A_{\epsilon_{\nu}^{\prime}}^{(n)}(p_{{0}} || p_{{1}})$ as follows:
\begin{align}
&A_{\epsilon_{\nu}^{\prime}}^{(n)} (p_{{0}} || p_{{1}}) \notag \\ 
&:= \left\{\, x^{n} \in \mathcal{X}^{n} \colon
\abs{\frac{1}{n} \log \frac{p_{X\mid \theta}(x^{n}\mid \theta_{0})}{p_{X\mid \theta}(x^{n}\mid \theta_{1})} - D(p_{0} || p_{1})} \leq  \epsilon_{\nu}^{\prime} \right\}. 
\end{align}
Moreover, define a sequence of test functions $(\delta_{n}^{*, \text{AEP}})_{n=1}^{\infty}$ as follows: 
\begin{align}
\delta_{n}^{*, \text{AEP}}(x^{n}, 1) &= 
\begin{cases}
0, & x^{n}\in A_{\epsilon_{\nu}^{\prime}}^{(n)}(p_{{0}} || p_{{1}}), \\ 
1, & \text{otherwise}. 
\end{cases}
\end{align}

Then, from the asymptotic equipartition property (AEP, see, \cite[Thm 11.8.2]{Cover:2006:EIT:1146355}), 
the following holds for sufficiently large $n$: 
\begin{align} 
&\alpha_{{\nu}}(\theta_{0}, \delta_{n}^{*, \text{AEP}})  \notag \\ 
&= \frac{\nu}{\nu-1} \Bigl\{1 - \sum_{x^{n}}p_{X^{n}\mid \theta}(x^{n}\mid \theta_{0}) \delta^{*, \text{AEP}}(x^{n}, 0)^{\frac{\nu-1}{\nu}} \Bigr\} \\ 
&= \frac{\nu}{\nu-1} \cdot \underbrace{\sum_{x^{n}\in A_{\epsilon_{\nu}^{\prime}}^{(n)^{c}}}(p_{0} || p_{1}) p_{X^{n}\mid\theta}(x^{n}\mid \theta_{0})}_{\leq \epsilon_{\nu}^{\prime}} \\ 
&\leq \frac{\nu}{\nu-1}\cdot \epsilon_{\nu}^{\prime} = \epsilon^{\prime}\leq \epsilon.
\end{align}

Similarly, it also holds from the AEP that 
\begin{align}
&-\lim_{n\to \infty}\frac{1}{n} \log \bar{\beta}_{{\nu}}(\theta_{1}, \delta_{n}^{*, \text{AEP}}) \notag \\
&= -\lim_{n\to \infty} \frac{1}{n}\log \frac{\nu}{\nu-1} \cdot \underbrace{\sum_{x^{n}\in A_{\epsilon_{\nu}^{\prime}}^{(n)}(p_{0} || p_{1})} p_{X^{n}\mid\theta}(x^{n}\mid \theta_{1})}_{<2^{-n(D(p_{0} || p_{1}) - \epsilon_{\nu}^{\prime})}}\\ 
&\geq D(p_{0} || p_{1}) - \epsilon_{\nu}^{\prime} > D(p_{0} || p_{1})-\epsilon^{\prime}. 
\end{align}
Therefore, $R = D(p_{0} || p_{1})-\epsilon^{\prime}$ is achievable. 
Since $\epsilon^{\prime}>0$ is arbitrary, we can conclude that 
$B_{\nu, \epsilon}\geq D(p_{0} || p_{1})$. 

\noindent
(Converse part): Let $\epsilon^{\prime}>0$ be an arbitrarily small number such that 
$0 < \epsilon^{\prime} < \epsilon$ and 
$(\delta^{*}_{n})_{n=1}^{\infty}$ be an arbitrary sequence of 
randomized test functions such that 
$\alpha_{{\nu}}(\theta_{0}, \delta^{*}_{n}) \leq \epsilon^{\prime}$ 
for sufficiently large $n$. Let $A_{\epsilon^{\prime}}^{(n)}(p_{0} || p_{1})$ be an $\epsilon^{\prime}$-relative typical sequences. 
First, we will show the next lemma by using the Jensen's inequality.
\begin{lemma} \label{lem:1-2epsilon_prime}
\begin{align}
\sum_{x^{n}\in A_{\epsilon^{\prime}}^{(n)}(p_{0} || p_{1})} p_{X^{n}\mid \theta}(x^{n}\mid \theta_{0})\delta^{*}_{n}(x^{n}, 0) \geq 1-2\epsilon^{\prime}. 
\end{align}
\end{lemma}
\begin{proof}
Since 
\begin{align}
\alpha_{{\nu}}(\theta_{0}, \delta^{*}_{n}) 
&:= \frac{\nu}{\nu-1} \left\{ 1 - \vE_{X^{n}\mid \theta_{0}}\left[\delta^{*}_{n}(X^{n}, 0)^{\frac{\nu-1}{\nu}}\right]\right\} \leq \epsilon^{\prime}
\end{align}
for sufficiently large $n$, it holds that 
\begin{align}
\vE_{X^{n}\mid \theta_{0}}\left[\delta^{*}_{n}(X^{n}, 0)^{\frac{\nu-1}{\nu}}\right] \geq 1 - \frac{\nu-1}{\nu}\cdot \epsilon^{\prime}. 
\end{align}

Then, it follows from the Jensen's inequality\footnote{Note that 
$\vE[Z]^{\frac{\nu-1}{\nu}} \geq \vE[Z^{\frac{\nu-1}{\nu}}]$ since 
$f(Z)=Z^{\frac{\nu-1}{\nu}}$ is a concave function for $0\leq Z \leq 1$.} 
that 
\begin{align}
\left\{\vE_{X^{n}\mid \theta_{0}}\left[\delta^{*}_{n}(X^{n}, 0)\right] \right\}^{\frac{\nu-1}{\nu}} 
&\geq \vE_{X^{n}\mid \theta_{0}}\left[\delta^{*}_{n}(X^{n}, 0)^{\frac{\nu-1}{\nu}}\right] \\ 
&\geq 1 - \frac{\nu-1}{\nu}\cdot \epsilon^{\prime}.
\end{align}

Thus, we have
\begin{align}
\vE_{X^{n}\mid \theta_{0}}\left[\delta^{*}_{n}(X^{n}, 0)\right] &\geq \left( 1 - \frac{\nu-1}{\nu}\cdot \epsilon^{\prime} \right)^{\frac{\nu}{\nu-1}} \\ 
&\overset{(a)}{\geq} 1-\epsilon^{\prime}, 
\end{align}
where 
\begin{itemize}
\item $(a)$ follows from $(1-\gamma x)^{1/\gamma} \geq 1-x$ for $0<\gamma\leq 1, 0\leq x\leq 1$. 
\end{itemize}

From the AEP, $\vE_{X^{n}\mid \theta_{0}}\left[\delta^{*}_{n}(X^{n}, 0)\right]$ can be upper bounded as follows: 
\begin{align}
&\vE_{X^{n}\mid \theta_{0}}\left[\delta^{*}_{n}(X^{n}, 0)\right] \notag \\ 
&=  \sum_{x^{n}\in  A_{\epsilon^{\prime}}^{(n)}(p_{0} || p_{1})} p_{X^{n}\mid \theta}(x^{n}\mid \theta_{0})\delta^{*}_{n}(x^{n}, 0) \notag \\ 
&\qquad + \sum_{x^{n}\in  A_{\epsilon^{\prime}}^{(n)^{c}}(p_{0} || p_{1})} p_{X^{n}\mid \theta}(x^{n}\mid \theta_{0})\delta^{*}_{n}(x^{n}, 0) \\ 
&\leq \sum_{x^{n}\in A_{\epsilon^{\prime}}^{(n)}(p_{0} || p_{1})} p_{X^{n}\mid \theta}(x^{n}\mid \theta_{0})\delta^{*}_{n}(x^{n}, 0) \notag \\ 
&\qquad + \underbrace{\sum_{x^{n}\in A_{\epsilon^{\prime}}^{(n)^{c}}(p_{0} || p_{1})} p_{X^{n}\mid \theta}(x^{n}\mid \theta_{0})}_{\leq \epsilon^{\prime}} \\ 
&\leq \sum_{x^{n}\in A_{\epsilon^{\prime}}^{(n)}(p_{0} || p_{1})} p_{X^{n}\mid \theta}(x^{n}\mid \theta_{0})\delta^{*}_{n}(x^{n}, 0) + \epsilon^{\prime}. 
\end{align}
Therefore, $\sum_{x^{n}\in A_{\epsilon^{\prime}}^{(n)}(p_{0} || p_{1})} p_{X^{n}\mid \theta}(x^{n}\mid \theta_{0})\delta^{*}_{n}(x^{n}, 0) \geq 1-2\epsilon^{\prime}$. 
\end{proof}

Next, $\bar{\beta}_{{\nu}}(\theta_{1}, \delta^{*}_{n})$ can be lower bounded as follows:
\begin{align}
&\bar{\beta}_{{\nu}}(\theta_{1}, \delta^{*}_{n}) \notag \\ 
&:= \frac{\nu}{\nu-1} \left\{1 - \sum_{x^{n}} p_{X^{n}\mid \theta}(x^{n}\mid \theta_{1})\delta^{*}_{n}(x^{n}, 1)^{\frac{\nu-1}{\nu}} \right\} \\ 
&= \frac{\nu}{\nu-1} \left\{1 - \sum_{x^{n}} p_{X^{n}\mid \theta}(x^{n}\mid \theta_{1})(1-\delta^{*}_{n}(x^{n}, 0))^{\frac{\nu-1}{\nu}}\right\} \\ 
&\overset{(b)}{\geq} \frac{\nu}{\nu-1} \left\{1 + \sum_{x^{n}}p_{X^{n}\mid \theta}(x^{n}|\theta_{1}) \left( -1 + \frac{\nu-1}{\nu} \delta^{*}_{n}(x^{n}, 0) \right) \right\} \\ 
&= \sum_{x^{n}} p_{X^{n}\mid \theta}(x^{n}\mid \theta_{1})\delta^{*}_{n}(x^{n}, 0), 
\end{align}
where 
\begin{itemize}
\item $(b)$ follows from $(1-x)^{\gamma} \leq 1-\gamma x$, for $0<\gamma\leq 1, 0\leq x\leq 1$. 
\end{itemize}

Making use of the result in Lemma \ref{lem:1-2epsilon_prime} and AEP, we have 
\begin{align*}
&\sum_{x^{n}} p_{X^{n}\mid \theta}(x^{n}\mid \theta_{1})\delta^{*}_{n}(x^{n}, 0) \\ 
&\geq \sum_{x^{n}\in  A_{\epsilon^{\prime}}^{(n)}(p_{0} || p_{1})} \underbrace{p_{X^{n}\mid \theta}(x^{n}\mid \theta_{1})}_{\qquad \geq p_{X^{n}\mid \theta}(x^{n}\mid \theta_{0})2^{-n (D(p_{0}||p_{1}) + \epsilon^{\prime})}}\delta^{*}_{n}(x^{n}, 0) \\ 
&\geq 2^{-n (D(p_{0} || p_{1}) + \epsilon^{\prime})} \sum_{x^{n}\in A_{\epsilon^{\prime}}^{(n)}(p_{0} || p_{1})} p_{x^{n}\mid \theta}(x^{n}\mid \theta_{0})\delta^{*}_{n}(x^{n}, 0) \\ 
&\geq 2^{-n (D(p_{0} || p_{1}) + \epsilon^{\prime})}(1-2\epsilon^{\prime}).
\end{align*}
Therefore, we have $-\liminf_{n\to \infty} \frac{1}{n}\log \bar{\beta}_{{\nu}}(\theta_{1}, \delta^{*}_{n}) \leq D(p_{0} || p_{1}) + \epsilon^{\prime}$.   
Since $\epsilon^{\prime}>0$ is arbitrary, we can conclude that $B_{\nu, \epsilon}\leq D(p_{0} || p_{1})$.
\end{proof}

\section{Proof of Theorem \ref{thm:nu_bayes_opt_exponent}} \label{proof:nu_bayes_opt_exponent}

\begin{proof}
From Proposition \ref{prop:opt_bayes_ht}, for $\nu = \infty$, the $\infty$-Bayesian error probability is given as 
\begin{align}
&r(\delta_{n}^{*, \text{Bayes}}) = \pi_{0}\sum_{x^{n}\in \textsf{A}(\delta_{n}^{*, \text{Bayes}})}p_{X^{n}\mid \theta}(x^{n}\mid \theta_{0}) \notag \\ 
&\qquad + \pi_{1}\sum_{x^{n}\in \textsf{A}(\delta_{n}^{*, \text{Bayes}})^{c}}p_{X^{n}\mid \theta}(x^{n}\mid \theta_{1}),
\end{align}
where $\textsf{A}(\delta_{n}^{*, \text{Bayes}}):= \left\{\, x^{n}\in \mathcal{X}^{n} \relmiddle{|} \pi(\theta_{0}\mid x^{n})\geq \pi(\theta_{1}\mid x^{n}) \right\}$. 
Thus the problem of characterizing $\infty$-Bayesian error exponent  $D_{\infty}^{*}$ is 
equivalent to the classical problem of characterizing the Bayesian error exponent $D^{*}$ defined in \eqref{eq:Bayesian_error_exponent}. 
Therefore, $D_{\infty}^{*} = C(p_{X\mid \theta_{0}}, p_{X\mid \theta_{1}})$ (see \cite{Cover:2006:EIT:1146355}). 

For $\nu\in (1, \infty)$, the $\nu$-Bayesian error probability is given, and upper bounded as follows: 
\begin{align}
&r(\delta_{n}^{*, \text{Bayes}})= \pi_{0}\cdot \frac{\nu}{\nu-1} \cdot \sum_{x^{n}}p_{X^{n}\mid \theta}(x^{n}\mid \theta_{0})  \notag \\
&\times \left\{1 - \left( \frac{\pi(\theta_{0}\mid x^{n})^{\nu}}{\pi(\theta_{0}\mid x^{n})^{\nu} + \pi(\theta_{1}\mid x^{n})^{\nu}} \right)^{\frac{\nu-1}{\nu}} \right\} \notag \\ 
&+ \pi_{1}\cdot \frac{\nu}{\nu-1}\cdot \sum_{x^{n}}p_{X^{n}\mid \theta}(x^{n}\mid \theta_{1}) \notag \\ 
&\times \left\{1 - \left( \frac{\pi(\theta_{1}\mid x^{n})^{\nu}}{\pi(\theta_{0}\mid x^{n})^{\nu} + \pi(\theta_{1}\mid x^{n})^{\nu}} \right)^{\frac{\nu-1}{\nu}} \right\} \\ 
&\overset{(a)}{\leq} \pi_{0}\cdot \frac{\nu}{\nu-1} \cdot \sum_{x^{n}}p_{X^{n}\mid \theta}(x^{n}\mid \theta_{0})  \notag \\
&\times \left\{1 - \frac{\pi(\theta_{0}\mid x^{n})^{\nu}}{\pi(\theta_{0}\mid x^{n})^{\nu} + \pi(\theta_{1}\mid x^{n})^{\nu}}  \right\} \notag \\ 
&+ \pi_{1}\cdot \frac{\nu}{\nu-1}\cdot \sum_{x^{n}}p_{X^{n}\mid \theta}(x^{n}\mid \theta_{1}) \notag \\ 
&\times \left\{1 - \frac{\pi(\theta_{1}\mid x^{n})^{\nu}}{\pi(\theta_{0}\mid x^{n})^{\nu} + \pi(\theta_{1}\mid x^{n})^{\nu}}  \right\} \\ 
&= \frac{\nu}{\nu-1}\cdot \sum_{x^{n}} \frac{\pi_{0}p_{X^{n}\mid \theta}(x^{n}\mid \theta_{0})\pi_{1}^{\nu} p_{X^{n}\mid \theta}(x^{n}\mid \theta_{1})^{\nu}}{\pi_{0}^{\nu} p_{X^{n}\mid \theta}(x^{n}\mid \theta_{0})^{\nu}+\pi_{1}^{\nu} p_{X^{n}\mid \theta}(x^{n}\mid \theta_{1})^{\nu}} \notag \\ 
&+ \frac{\nu}{\nu-1}\cdot  \sum_{x^{n}} \frac{\pi_{1}p_{X^{n}\mid \theta}(x^{n}\mid \theta_{1})\pi_{0}^{\nu} p_{X^{n}\mid \theta}(x^{n}\mid \theta_{0})^{\nu}}{\pi_{0}^{\nu} p_{X^{n}\mid \theta}(x^{n}\mid \theta_{0})^{\nu}+\pi_{1}^{\nu} p_{X^{n}\mid \theta}(x^{n}\mid \theta_{1})^{\nu}} \\ 
&= \frac{\nu}{\nu-1}\cdot \sum_{x^{n}} \frac{1}{\frac{\pi_{0}^{\nu-1} p_{X^{n}\mid \theta}(x^{n}\mid \theta_{0})^{\nu-1}}{\pi_{1}^{\nu}p_{X^{n}\mid \theta}(x^{n}\mid \theta_{1})^{\nu}} + \frac{1}{\pi_{0}p_{X^{n}\mid \theta}(x^{n}\mid \theta_{0})}} \notag \\ 
&+ \frac{\nu}{\nu-1}\cdot \sum_{x^{n}} \frac{1}{\frac{\pi_{1}^{\nu-1} p_{X^{n}\mid \theta}(x^{n}\mid \theta_{1})^{\nu-1}}{\pi_{0}^{\nu}p_{X^{n}\mid \theta}(x^{n}\mid \theta_{0})^{\nu}} + \frac{1}{\pi_{1}p_{X^{n}\mid \theta}(x^{n}\mid \theta_{1})}} \\ 
&\overset{(b)}{\leq} \frac{1}{2}\cdot \frac{\nu}{\nu-1} \sum_{x^{n}} (\pi_{1}p_{X^{n}\mid\theta}(x^{n}\mid \theta_{1}))^{\frac{\nu}{2}} (\pi_{0}p_{X^{n}\mid \theta}(x^{n}\mid \theta_{0}))^{1-\frac{\nu}{2}} \notag \\ 
&+ \frac{1}{2}\cdot \frac{\nu}{\nu-1} \sum_{x^{n}} (\pi_{0}p_{X^{n}\mid\theta}(x^{n}\mid \theta_{0}))^{\frac{\nu}{2}} (\pi_{1}p_{X^{n}\mid \theta}(x^{n}\mid \theta_{1}))^{1-\frac{\nu}{2}} \\ 
&= \frac{1}{2}\cdot \frac{\nu}{\nu-1}\cdot \pi_{1}^{\frac{\nu}{2}}\pi_{0}^{1-\frac{\nu}{2}}\sum_{x^{n}} \prod_{i=1}^{n}p_{X\mid \theta}(x_{i}\mid \theta_{1})^{\frac{\nu}{2}} p_{X\mid \theta}(x_{i}\mid \theta_{0})^{1-\frac{\nu}{2}} \notag \\ 
&+ \frac{1}{2}\cdot \frac{\nu}{\nu-1}\cdot \pi_{0}^{\frac{\nu}{2}}\pi_{1}^{1-\frac{\nu}{2}}\sum_{x^{n}} \prod_{i=1}^{n}p_{X\mid \theta}(x_{i}\mid \theta_{0})^{\frac{\nu}{2}} p_{X\mid \theta}(x_{i}\mid \theta_{1})^{1-\frac{\nu}{2}} \\ 
&= \frac{1}{2}\cdot \frac{\nu}{\nu-1}\cdot \pi_{1}^{\frac{\nu}{2}}\pi_{0}^{1-\frac{\nu}{2}} 
\left( \sum_{x} p_{X\mid \theta}(x\mid \theta_{1})^{\frac{\nu}{2}}p_{X\mid \theta}(x\mid \theta_{0})^{1-\frac{\nu}{2}} \right)^{n} \notag \\ 
&+ \frac{1}{2}\cdot \frac{\nu}{\nu-1}\cdot \pi_{0}^{\frac{\nu}{2}}\pi_{1}^{1-\frac{\nu}{2}} 
\left( \sum_{x} p_{X\mid \theta}(x\mid \theta_{0})^{\frac{\nu}{2}}p_{X\mid \theta}(x\mid \theta_{1})^{1-\frac{\nu}{2}} \right)^{n} \\ 
&\overset{(c)}{\leq} \frac{\nu}{\nu-1} \max \Bigl\{\sum_{x} p_{X\mid \theta}(x\mid \theta_{1})^{\frac{\nu}{2}}p_{X\mid \theta}(x\mid \theta_{0})^{1-\frac{\nu}{2}}, \notag \\  
&\qquad \qquad  \qquad \sum_{x} p_{X\mid \theta}(x\mid \theta_{0})^{\frac{\nu}{2}}p_{X\mid \theta}(x\mid \theta_{1})^{1-\frac{\nu}{2}} \Bigr\}^{n}, 
\end{align}
where 
\begin{itemize}
\item $(a)$ follows from $1-x^{\gamma} \leq 1-x$ for $0\leq\gamma\leq1$, 
\item $(b)$ follows from $\text{(harmonic mean)}\leq \text{(geometric mean)}$,
\item $(c)$ follows from $a^{n} + b^{n} \leq 2\max\{a, b\}^{n}$ for $a,b\geq 0$ and $\pi_{0}, \pi_{1}\leq 1$.
\end{itemize}
Therefore, 
\begin{align}
D_{\nu}^{*} &= -\lim_{n\to \infty} \frac{1}{n} \log r(\delta_{n}^{*, \text{Bayes}}) \\ 
&\geq -\log \max\{\text{BC}_{{\nu}/{2}}, \text{BC}_{1-{\nu}/{2}}\}.
\end{align}
For $\nu=1$, it can be proved in a similar way as $\nu\in (1, \infty)$ by using the inequality 
$1-1/x\leq \log x$.
\end{proof}


\end{document}